%% file: maxcut.tex
\documentclass[a4paper,11pt]{article}

\usepackage[utf8]{inputenc}
\usepackage{amsmath,amsthm,amssymb}
\usepackage{mathtools}
\usepackage{tikz-cd}
\usepackage{graphicx}
\usepackage{xcolor}
\usepackage[font=small]{caption}
\usepackage{algorithm}
\usepackage[noend]{algpseudocode}
\usepackage{setspace}
\usepackage[numbers,sort&compress]{natbib}
\usepackage{fullpage}
\usepackage[colorlinks=true,linkcolor=blue,citecolor=red]{hyperref}
\usepackage[nameinlink]{cleveref}

\algrenewcommand\alglinenumber[1]{{\sf\footnotesize#1}}
\let\Algorithm\algorithm
\renewcommand\algorithm[1][]{\Algorithm[#1]\setstretch{1.15}}
\floatname{algorithm}{Procedure}

\Crefname{algorithm}{Procedure}{Procedures}
\Crefname{equation}{Inequality}{Inequalities}
\creflabelformat{equation}{#2\textup{#1}#3}

\usetikzlibrary{fit}
\tikzset{%
  dot/.style={circle, inner sep=0pt, minimum size=2mm, fill=black}
}
\tikzset{%
  dotlabel/.style={circle, thick, draw, inner sep=1pt, minimum size=9pt}
}

\newtheorem{theorem}{Theorem}
\newtheorem{lemma}{Lemma}
\newtheorem{proposition}{Proposition}

\newcommand{\poly}{\text{poly}}
\newcommand{\etal}{et al.}
\DeclarePairedDelimiter\ceil{\lceil}{\rceil}
\DeclarePairedDelimiter\floor{\lfloor}{\rfloor}

\newcommand{\problemdef}[3]{%
  \medskip
  \hspace{-1em}
    \fbox { \parbox { .9\textwidth} {\vspace{0.1cm}
    \textsc{{\hspace{-0.5em} \textsc{#1}}}
    \vspace{-0.2cm}
    \begin{description}
        \item[Input: ] #2
        \vspace{-0.2cm}
        \item[Task:]\hspace{1mm} #3
    \end{description}
    \vspace{-0.2cm}}}
    \medskip
}

\title{Exact Algorithms for MaxCut on Split Graphs}
\author{Marko Lalovic\footnote{Hamburg University of Technology, Germany. \texttt{marko.lalovic@tuhh.de}.}}
\date{}

\begin{document}

\maketitle

\begin{abstract}
  This paper presents an $O^{*}(1.42^{n})$ time algorithm for the Maximum Cut problem on split graphs, along with a subexponential time algorithm for its decision variant.
\end{abstract}


\section{Introduction}
\label{sec:introduction}%

The \emph{Maximum Cut problem}, MaxCut for short, asks to partition a graph's vertices into two sets to maximize the \emph{cut size}, defined as the number of edges between the sets. Variants of MaxCut have applications in various fields, including social networks~\cite{Harary1959}, data clustering~\cite{Poland2006}, and image segmentation~\cite{Sousa2013}. Efficiently solving MaxCut is essential for these applications.

The \emph{decision variant of MaxCut} (given a graph $G$ and $k \in \mathbb{N}$, is there a cut of size at least $k$ in $G$?) was proven to be NP-complete by Garey et al.~\cite{Garey1976}. The fastest algorithm for solving MaxCut, proposed by Williams~\cite{Williams2004}, has a running time of approximately $O^{*}(1.74^n)$~\footnote{The $O^{*}$ notation suppresses polynomially bounded terms. For a positive real constant $c$, we write $O^{*}(c^{n})$ for a time complexity of $O(c^{n} \cdot \poly)$ in the number of vertices $n$. For a detailed discussion of this, see the survey by Woeginger~\cite{Woeginger2003}.}. However, this algorithm requires exponential space. Whether a polynomial space algorithm for MaxCut exists that runs faster than $O^{*}(2^{n})$ is an open problem listed in~\cite{Woeginger2008}. This situation improves when the input is restricted to special classes of graphs.

A split graph is one whose vertex set can be partitioned into a clique and an independent set. This places split graphs "halfway" between bipartite graphs and their complements, making the study of MaxCut on split graphs particularly interesting. Moreover, Bodlaender and Jansen~\cite{Bodlaender2000} proved that decision variant of MaxCut remains NP-complete on split graphs. Sucupira \etal\ \cite{Sucupira2014} provided a polynomial time algorithm for MaxCut on a special subclass called full $(k,n)$-split graphs. More recently, Bliznets and Epifanov~\cite{Bliznets2023} studied MaxCut parameterized above the spanning tree bound and provided subexponential time algorithms for this variant of MaxCut on chordal, co-bipartite, and split graphs.

This paper presents an $O^{*}(1.42^{n})$ time algorithm for MaxCut on split graphs, along with a subexponential time algorithm for the decision variant of MaxCut that is asymptotically optimal under the ETH. Both algorithms use polynomial space. \Cref{sec:preliminaries} covers the notation and preliminaries, including conditional lower bounds based on the ETH. \Cref{sec:expo-algo,sec:subexpo-algo} describe the new algorithms, along with proofs of their correctness and run-time guarantees. In Conclusions, \Cref{sec:conclusions} outlines several promising paths for some follow-up work.

\section{Preliminaries}
\label{sec:preliminaries}%

Let $G = \left(V, E\right)$ be a simple graph on $n$ vertices. Complement of the graph $G$ is denoted by $\overline{G}$. The neighborhood of a vertex $v$ is denoted by $N(v) = \lbrace u \in V : (u, v) \in E \rbrace$. A subset $C \subseteq V$ is called a {\it clique}, if every two vertices in $C$ are adjacent. A subset $I \subseteq V$ is called an {\it independent set}, if no two vertices in $I$ are adjacent. A partition of vertices $S \subseteq V$ into two disjoint subsets $S_{1}$ and $S_{2}$ is denoted by $(S_{1}, S_{2})$. A partition $(V_{1}, V_{2})$ of $V$ is called a \emph{cut}. The subset of edges $E$ that have one endpoint in $V_{1}$ and the other endpoint in $V_{2}$ is denoted by $E(V_{1}, V_{2})$. The \emph{size} of a cut is the cardinality of $E(V_{1}, V_{2})$ denoted by $k$:
$$
k = |E(V_{1}, V_{2})|.
$$

The MaxCut problem asks for a cut of maximum size in a given graph $G$. In the decision variant of MaxCut one is given a graph $G$ and an integer $k$, and the question is whether $G$ has a cut of size at least $k$. Formally, the problems are defined as follows:

\problemdef{MaxCut}
{A graph $G$.}
{Find a cut of maximum size in $G$.}

\problemdef{MaxCut (decision variant)}
{A graph $G$ and an integer $k$.}
{Does $G$ have a cut of size $\geq k$?}

A graph $G = (V, E)$ is a split graph if $V$ can be partitioned into a clique $C$ and an independent set $I$. Such a partition can be found in linear time, as shown in~\cite{Golumbic2004}. If $G$ is a disconnected graph, then a maximum cut can be determined by solving MaxCut on each of its connected components individually and combining the results. If $G$ is a complete graph or an empty graph, then MaxCut on $G$ is trivial. Therefore, when solving MaxCut, we may assume that the input graph $G$ is a connected split graph with a partition of its vertex set into a clique $C$ and an independent set $I$, where $|C| \geq 1$ and $|I| \geq 1$.

Bodlaender and Jansen~\cite{Bodlaender2000} introduced a reduction of MaxCut on an arbitrary graph $G$ to MaxCut on a split graph $G'$. This reduction transforms graph $G$ into split graph $G'$ with $|V| + |E(\overline{G})|$ vertices and $|V| \cdot (|V| - 1) / 2 + 2|E(\overline{G})|$ edges. Moreover, there is a maximum cut of size $k$ in $G$ if and only if there is a maximum cut of size $k + 2|E(\overline{G})|$ in $G'$. \Cref{fig:graph-G} shows a graph $G$ with $|E(\overline{G})| = 4$ and a cut $(\lbrace v_1, v_3, v_5 \rbrace, \lbrace v_2, v_4 \rbrace)$ of size 5. The corresponding split graph $G'$, shown in \Cref{fig:split-graph-G-prime}, has a cut of size $5 + 2 \cdot 4 = 13$.

\begin{figure}
\centering
\begin{minipage}{.5\textwidth}
  \centering
  \input{graph-G.tex}
  \captionof{figure}{A graph $G$.}
  \label{fig:graph-G}%
\end{minipage}%
\begin{minipage}{.5\textwidth}
  \centering
  \input{split-graph-G-prime.tex}
  \captionof{figure}{A split graph $G'$.}
  \label{fig:split-graph-G-prime}%
\end{minipage}
\end{figure}

Under the Exponential Time Hypothesis (ETH), MaxCut cannot be solved in $2^{o(n)} \, \poly$ time on an arbitrary graph $G$, as proved in~\cite{Garey1976}. Using this foundational result along with the described reduction, we can establish the following lower bounds for MaxCut on split graphs.

\begin{proposition}\label{prop: lower-bounds}
There are no algorithms for MaxCut or its decision variant on split graphs with a running time of\, $2^{o(\sqrt{n})} \, \poly$ or\, $2^{o(\sqrt{k})} \, \poly$ unless ETH fails.
\end{proposition}
\begin{proof}
Let $G = (V, E)$ be an arbitrary graph with $n$ vertices and a maximum cut of size $k$. Using the reduction by Bodlaender and Jansen~\cite{Bodlaender2000}, we can construct the corresponding split graph $G'$ in polynomial time. The split graph $G'$ will have
$$
n' = |V(G)| + |E(\overline{G})|
$$ vertices and a maximum cut of size 
$$
k' = k + 2|E(\overline{G})|.
$$

Note that $n' = O(n^2)$ and $k' = O(n^2)$. Therefore, if we could solve MaxCut on $G'$ in $2^{o(\sqrt{n'})} \, \poly$ or $2^{o(\sqrt{k'})} \, \poly$ time, we would be able to solve MaxCut on the original graph $G$ in $2^{o(n)} \, \poly$ time. However, under the ETH assumption, MaxCut does not admit a $2^{o(n)} \, \poly$ algorithm on an arbitrary graph $G$.
\end{proof}

\section{Exponential Algorithm}
\label{sec:expo-algo}%

This section provides an exponential time algorithm for MaxCut on split graphs. First, we describe two procedures for MaxCut that are used as subroutines and provide proofs of correctness. The algorithm presented in the next section is also based on these results.

\begin{figure}
\noindent
\begin{minipage}[t]{0.55\textwidth}
\begin{algorithm}[H]
\begin{algorithmic}[1]
\Require Graph $G = (V, E)$, independent set $I \subseteq V$
\Ensure Maximum cut $(V_{1}, V_{2})$ of $G$
\State $C \gets V \setminus I$
\State $maxCutSize \gets 0$
\ForAll{$C_{1} \subseteq C$}
  \State $C_{2} \gets C \setminus C_{1}$
  \State $I_{1} \gets \lbrace v \in I : |N(v) \cap C_{2}| \geq |N(v) \cap C_{1}| \rbrace$
  \State $I_{2} \gets I \setminus I_{1}$
  \State $cutSize \gets |E(C_{1} \cup I_{1}, C_{2} \cup I_{2})|$
  \If{$cutSize \geq maxCutSize$}
    \State $(V_{1}, V_{2}) \gets (C_{1} \cup I_{1}, C_{2} \cup I_{2})$
    \State $maxCutSize \gets cutSize$
  \EndIf
\EndFor
\State \textbf{return} $(V_{1}, V_{2})$
\end{algorithmic}
\caption{\textproc{MaxCut}$(G, I)$}\label{algo:proc1}
\end{algorithm}
\end{minipage}%
\hspace{2em}%
\begin{minipage}[t]{0.4\textwidth}\vspace{40pt}
\centering
\input{split-graph-1.tex}
\captionof{figure}{A split graph.}
\label{fig:graph-1}%
\end{minipage}
\end{figure}

Given a graph $G = (V, E)$ and an independent set $I \subseteq V$, let $(C_1, C_2)$ be a partition of $C = V \setminus I$. Define $I_1$ as the set of vertices in $I$ that have more neighbors in $C_2$ than in $C_1$, and let $I_2 = I \setminus I_1$. This partition $(I_1, I_2)$ of $I$ maximizes the size of the cut $(C_1 \cup I_1, C_2 \cup I_2)$. An example of a split graph with a maximum cut $(C_1 \cup I_1, C_2 \cup I_2)$ of size 14 is shown in \Cref{fig:graph-1}. By evaluating all possible partitions $(C_1, C_2)$ of $C$ we can find a maximum cut $(V_1, V_2)$ of $G$. Refer to \Cref{algo:proc1} and \Cref{lem:correctness1}.

\begin{lemma}\label{lem:correctness1} Given a graph $G = (V, E)$ and an independent set $I \subseteq V$, \Cref{algo:proc1} returns a maximum cut $(V_1, V_2)$ of $G$.
\end{lemma}
\begin{proof}
Let $(V_{1}', V_{2}')$ be a cut in $G$. Let $C = V \setminus I$. \Cref{algo:proc1} considers all subsets of $C$ including the partition
$$ (C_{1}, C_{2}) = (V_{1}' \cap C, V_{2}' \cap C). $$

For each vertex $v \in I$, the number of edges in the cut is at most
$$ \max \lbrace |N(v) \cap C_{1}|, |N(v) \cap C_{2}| \rbrace. $$

If $|N(v) \cap C_{1}| \geq |N(v) \cap C_{2}|$, $v$ is assigned to $I_{1}$, otherwise, $v$ is assigned to $I_{2}$. Consequently,
$$ |E(C_{1} \cup I_{1}, C_{2} \cup I_{2})| \geq |E(V_{1}', V_{2}')|. $$
\end{proof}

\begin{figure}
\noindent
\begin{minipage}[t]{0.55\textwidth}
\begin{algorithm}[H]
\begin{algorithmic}[1]
\Require Graph $G = (V, E)$, clique $C \subseteq V$
\Ensure Maximum cut $(V_{1}, V_{2})$ of $G$
\State $I \gets V \setminus C$
\State $maxCutSize \gets 0$
\ForAll{$I_{1} \subseteq I$}
  \State $I_{2} \gets I \setminus I_{1}$
  \State $\textsc{Sort}\left(C, |N(v) \cap I_{2}| - |N(v) \cap I_{1}|\right)$
  \ForAll{$m \in \lbrace 0,1,2, \dots, |C| \rbrace$}
    \State $C_{1} \gets \lbrace v_{1}, v_{2}, \dots, v_{m} \rbrace$
    \State $C_{2} \gets C \setminus C_{1}$
    \State $cutSize \gets |E(C_{1} \cup I_{1}, C_{2} \cup I_{2})|$
    \If{$cutSize \geq maxCutSize$}
      \State $(V_{1}, V_{2}) \gets (C_{1} \cup I_{1}, C_{2} \cup I_{2})$
      \State $maxCutSize \gets cutSize$
    \EndIf
  \EndFor
\EndFor
\State \textbf{return} $(V_{1}, V_{2})$
\end{algorithmic}
\caption{\textproc{MaxCut}$(G, C)$}
\label{algo:proc2}%
\end{algorithm}
\end{minipage}%
\hspace{2em}%
\begin{minipage}[t]{0.4\textwidth}\vspace{55pt}
\centering
\input{split-graph-2.tex}
\captionof{figure}{A split graph.}
\label{fig:graph-2}%
\end{minipage}
\end{figure}

Given a graph $G = (V, E)$ and a clique $C \subseteq V$, let $(I_1, I_2)$ be a partition of $I = V \setminus C$. Define $C_1$ as the set of $m$ vertices $v$ in $C$ with the highest values of $ |N(v) \cap I_{2}| - |N(v) \cap I_{1}|,$ and let $C_2 = C \setminus C_1$. This ensures that vertices in $C_1$ have many neighbors in $I_2$ and vertices in $C_2$ have many neighbors in $I_1$. By considering all possible values of $m$, we can find a partition $(C_1, C_2)$ of $C$ that maximizes the size of the cut $(C_1 \cup I_1, C_2 \cup I_2).$ An example of a split graph with a maximum cut $(C_1 \cup I_1, C_2 \cup I_2)$ of size 14 is shown in \Cref{fig:graph-2}. In this example, $m = 2$, and vertices $v_{1} \in C$ are listed in non-increasing order based on $|N(v_{i}) \cap I_{2}| - |N(v_{i}) \cap I_{1}|$. By evaluating all possible partitions $(I_1, I_2)$ of $I$, we can find a maximum cut $(V_1, V_2)$ of $G$. Refer to \Cref{algo:proc2} and \Cref{lem:correctness2}, which shows that the cut produced by this procedure is at least as large as any other possible cut in $G$.

\begin{lemma}\label{lem:correctness2} Given a graph $G = (V, E)$ and a clique $C \subseteq V$, \Cref{algo:proc2} returns a maximum cut $(V_1, V_2)$ of $G$.
\end{lemma}
\begin{proof}
Let $(V_{1}', V_{2}')$ be a cut in $G$. Let $I = V \setminus C$. \Cref{algo:proc2} considers all subsets of $I$ including the partition
$$ (I_{1}, I_{2}) = (V_{1}' \cap I, V_{2}' \cap I). $$

It then sorts the vertices in $C$ in non-increasing order based on
$$ |N(v) \cap I_{2}| - |N(v) \cap I_{1}|. $$

Denote the vertices in $C$ in this order as:
$$ C = \lbrace v_{1}, v_{2}, \dots, v_{|C|} \rbrace. $$

Thus, for all $i = 1, \dots, |C| - 1$, we have:
\begin{equation}
\label{ineq:ordering}%
|N(v_{i}) \cap I_{2}| - |N(v_{i}) \cap I_{1}| \geq
|N(v_{i+1}) \cap I_{2}| - |N(v_{i+1}) \cap I_{1}|.
\end{equation}

Let $m = |V_{1}' \cap C|$. \Cref{algo:proc2} considers all possible values of $m$, so we can assume $m$ is known. It then selects the following partition of $C$:
$$ C_{1} = \lbrace v_{1}, v_{2}, \dots, v_{m} \rbrace, \quad C_{2} = C \setminus C_{1}. $$

This partition maximizes the size of the cut in $G$, which can be expressed as:
$$ |E(C_{1} \cup I_{1}, C_{2} \cup I_{2})| = m (|C| - m) +  \sum_{v \in C_{1}} |N(v) \cap I_{2}| + \sum_{v \in C_{2}} |N(v) \cap I_{1}|. $$

To prove this, recall that $m = |V_{1}' \cap C|$, and denote 
$$ (C_{1}', C_{2}') = (V_{1}' \cap C, V_{2}' \cap C).$$ 

From the ordering (\Cref{ineq:ordering}), we have:
$$
\sum_{v \in C_{1}} |N(v) \cap I_{2}| - |N(v) \cap I_{1}| \geq
\sum_{v \in C_{1}'} |N(v) \cap I_{2}| - |N(v) \cap I_{1}|.
$$
Adding $\sum_{v \in C} |N(v) \cap I_{1}|$ to both sides of this inequality, we get:
$$
\sum_{v \in C_{1}} |N(v) \cap I_{2}| + \sum_{v \in C_{2}} |N(v) \cap I_{1}|
\geq
\sum_{v \in C_{1}'} |N(v) \cap I_{2}| + \sum_{v \in C_{2}'} |N(v) \cap I_{1}|.
$$
Therefore, the size of the cut produced by \Cref{algo:proc2} satisfies:
$$ |E(C_{1} \cup I_{1}, C_{2} \cup I_{2})| \geq |E(V_{1}', V_{2}')|. $$
\end{proof}

\begin{theorem} MaxCut on split graphs can be solved in time 
$$ O\left(2^{n/2} \, \poly\right) = O^{*}(1.42^{n}). $$
\end{theorem}
\begin{proof}
Let $G = (V, E)$ be a split graph with a partition of its vertex set into a clique $C$ and an independent set $I$. Such a partition can be found in linear time, as shown in~\cite{Golumbic2004}.

For $|C| \leq n/2 $, by \Cref{lem:correctness1}, we can solve MaxCut using \Cref{algo:proc1} by considering all possible subsets of $C$ in time
$$ O\left(2^{|C|} \, \poly\right) = O\left(2^{n/2} \, \poly\right). $$

For $|C| > n/2$, by \Cref{lem:correctness2}, we can solve MaxCut using \Cref{algo:proc2} by considering all possible subsets of $I$ in time 
$$ O\left(2^{n - |C|} \, \poly \right) = O\left(2^{n/2} \, \poly\right). $$
\end{proof}

Note that neither \Cref{lem:correctness1} nor \Cref{lem:correctness2} assumes that the input graph $G$ is a split graph. Therefore, these lemmas can also be used to obtain faster algorithms for MaxCut on graphs with a large independent set or a large clique. Given a graph $G = (V, E)$ and a subset $S \subseteq V$, which is either a clique or an independent set, then MaxCut can be solved on $G$ in time $O\left(2^{|V \setminus S|} \, \poly \right).$

\section{Subexponential Algorithm}
\label{sec:subexpo-algo}%

This section describes a parameterized subexponential time algorithm for the decision variant of the MaxCut problem on split graphs. According to Proposition~\ref{prop: lower-bounds}, this algorithm is asymptotically optimal unless the ETH is false.

\begin{theorem}\label{thm:decision-maxcut}
Decision variant of MaxCut on split graphs can be solved in $2^{O(\sqrt{k})} \, \poly$ time.
\end{theorem}
\begin{proof}
Let $G = (V, E)$ be a split graph with a clique $C$. Consider a partition $(C_1, C_2)$ of $C$ such that
$$ |C_1| = \floor*{\frac{|C|}{2}} \quad \text{and} \quad |C_2| = \ceil*{\frac{|C|}{2}}.$$ 

Note that there are $|C_1| \cdot |C_2|$ edges in $G$ with one endpoint in $C_1$ and the other in $C_2$. Thus, we can immediately output "Yes", if
$$ k \leq \left(\frac{|C|}{2} \right)^{2}. $$

Otherwise, we have $|C| \leq 2 \sqrt{k}$. By \Cref{lem:correctness1}, the exact size of a maximum cut in $G$ can be found using \Cref{algo:proc1} by considering all possible subsets of C in time $ 2^{O(|C|)} \, \poly = 2^{O(\sqrt{k})} \, \poly.$
\end{proof}

\section{Conclusions}
\label{sec:conclusions}%
The algorithm presented in \Cref{sec:expo-algo} for MaxCut on split graphs can be adapted to create faster algorithms for graphs with large homogeneous parts. For instance, with some modifications, we can develop an algorithm for MaxCut on double split graphs with the same running time. Whether this approach can be extended to other graphs, such as chordal graphs, remains unclear. The presented subexponential algorithm in \Cref{sec:subexpo-algo} for the decision variant of MaxCut on split graphs is essentially tight under the ETH. A similar result to \Cref{thm:decision-maxcut} can be obtained for the parametric dual of MaxCut, known as \emph{Edge Bipartization problem}, on split graphs. 

It is important to note that the reduction by Bodlaender and Jansen~\cite{Bodlaender2000} involves a transformation that causes a quadratic increase in the number of vertices $n$. As a result, this reduction cannot be used to rule out a subexponential time algorithm in terms of $n$. Determining whether such an algorithm exists for MaxCut on split graphs remains an interesting open question.

One strategy for developing fast algorithms for problems on split graphs is to use random sampling instead of considering every possible subset of a clique or independent set, as done here. Combining the presented algorithms with polynomial time approximation schemes~\cite{Arora1999} and adapting techniques from~\cite{Lochet2021} for everywhere dense graphs to split graphs seems like a promising way to develop faster algorithms for problems on split graphs.

\medskip\noindent\textbf{Acknowledgments.} I would like to thank Prof.\ Dr.\ Matthias Mnich for suggesting me this topic and for his support in answering my questions and clarifying problems.

\bibliographystyle{abbrvnat}
\bibliography{maxcut}
\end{document}

%% file: graph-G.tex
\begin{tikzpicture}
\begin{scope}[fill opacity=0.5, text opacity=1]
    \node[dotlabel] (v1) at (0, 0) {\small $v_1$};
    \node[dotlabel, fill=lightgray] (v2) at (1.14, 0) {\small $v_2$};
    \node[dotlabel] (v3) at (1.49, 1.09) {\small $v_3$};
    \node[dotlabel, fill=lightgray] (v4) at (0.63, 1.83) {\small $v_4$};
    \node[dotlabel] (v5) at (-0.35, 1.09) {\small $v_5$};
    \draw[-, thick] (v1) -- (v2);
    \draw[-, thick] (v2) -- (v3);
    \draw[-, thick] (v3) -- (v4);
    \draw[-, thick] (v4) -- (v5);
    \draw[-, thick] (v5) -- (v1);
    \draw[-, thick] (v5) -- (v2);
\end{scope}
\end{tikzpicture}

%% file: split-graph-G-prime.tex
\begin{tikzpicture}
\begin{scope}[fill opacity=0.5, text opacity=1]
    \node[dotlabel] (v1) at (0, 0) {\small $v_1$};
    \node[dotlabel, fill=lightgray] (v2) at (1.14, 0) {\small $v_2$};
    \node[dotlabel] (v3) at (1.49, 1.09) {\small $v_3$};
    \node[dotlabel, fill=lightgray] (v4) at (0.63, 1.83) {\small $v_4$};
    \node[dotlabel] (v5) at (-0.35, 1.09) {\small $v_5$};
    \draw[-, thick] (v1) -- (v2); %
    \draw[-, thin] (v1) -- (v3);
    \draw[-, thin] (v1) -- (v4);
    \draw[-, thick] (v1) -- (v5); %
    \draw[-, thick] (v2) -- (v3); %
    \draw[-, thin] (v2) -- (v4);
    \draw[-, thick] (v2) -- (v5); %
    \draw[-, thick] (v3) -- (v4); %
    \draw[-, thin] (v3) -- (v5);
    \draw[-, thick] (v4) -- (v5); %
    
    \node[dotlabel, fill=lightgray] (e1) at (3.45, 3.5) {\small $v_9$};
    \node[dotlabel] (e3) at (3.45, 1.79) {\small $v_8$};    
    \node[dotlabel, fill=lightgray] (e2) at (3.45, 0.07) {\small $v_7$};
    \node[dotlabel, fill=lightgray] (e4) at (3.45, -1.64) {\small $v_6$};
    
    \draw[-, thin] (e1) -- (v1);
    \draw[-, thin] (e1) -- (v4);

    \draw[-, thin] (e2) -- (v3);
    \draw[-, thin] (e2) -- (v5);

    \draw[-, thin] (e3) -- (v2);
    \draw[-, thin] (e3) -- (v4);
    
    \draw[-, thin] (e4) -- (v1);
    \draw[-, thin] (e4) -- (v3);
\end{scope}
\end{tikzpicture}

%% file: split-graph-1.tex
\begin{tikzpicture}[xscale=1,yscale=1.2]
    \node[dot] (v1) at (0, 4) {};    
    \node[dot] (v2) at (0, 3) {};    
    \node[dot] (v3) at (0, 2) {};    
    \node[dot] (v4) at (0, 1) {};    
    \node[dot] (v5) at (0, 0) {};
    
    \node[dot] (u1) at (2, 4) {};    
    \node[dot] (u2) at (2, 3) {};
    \node[dot] (u3) at (2, 2) {};
    \node[dot] (u4) at (2, 1) {};
    \node[dot] (u5) at (2, 0) {};

    \node[fit=(v1) (v2), thick, draw=gray, minimum width=1.2cm] {};
    \node[left=.6cm, font=\boldmath] at (0, 3.5) {$C_{1}$};

    \node[fit=(v3) (v5), thick, dashed, draw=lightgray, minimum width=1.2cm] {};
    \node[left=.6cm, font=\boldmath, opacity=1] at (v4) {$C_{2}$};

    \node[fit=(u1) (u2), thick, dashed, draw=lightgray, opacity=0.8] {};
    \node[right=.25cm, font=\boldmath] at (2, 3.5) {$I_{2}$};
    
    \node[fit=(u3) (u5), thick, draw=gray] {};
    \node[right=.25cm, font=\boldmath] at (u4) {$I_{1}$};

    \draw[-] (v1) -- (v2);
    \draw[-] (v2) -- (v3);
    \draw[-] (v3) -- (v4);
    \draw[-] (v4) -- (v5);    
    \draw[-] (v1) to[bend right] (v3);
    \draw[-] (v1) to[bend right] (v4);
    \draw[-] (v1) to[bend right] (v5);
    \draw[-] (v2) to[bend right] (v4);
    \draw[-] (v2) to[bend right] (v5);
    \draw[-] (v3) to[bend right] (v5);

    \draw[-] (u1) -- (v1);
    \draw[-] (u2) -- (v1);
    \draw[-] (u2) -- (v2);
    \draw[-] (u2) -- (v3);
    \draw[-] (u3) -- (v4);
    \draw[-] (u3) -- (v5);
    \draw[-] (u4) -- (v5);
    \draw[-] (u5) -- (v3);
    \draw[-] (u5) -- (v5);  
\end{tikzpicture}

%% file: split-graph-2.tex
\begin{tikzpicture}[xscale=1,yscale=1.2]
    \node[dotlabel] (v1) at (0, 4) {\small $v_{1}$};    
    \node[dotlabel] (v2) at (0, 3) {\small $v_{2}$};        
    \node[dotlabel] (v3) at (0, 2) {\small $v_{3}$}; 
    \node[dotlabel] (v4) at (0, 1) {\small $v_{4}$};
    \node[dotlabel] (v5) at (0, 0) {\small $v_{5}$};
    
    \node[dot] (u1) at (2, 4) {};    
    \node[dot] (u2) at (2, 3) {};
    \node[dot] (u3) at (2, 2) {};
    \node[dot] (u4) at (2, 1) {};
    \node[dot] (u5) at (2, 0) {};

    \node[fit=(v1) (v2), thick, draw=gray, minimum width=1.2cm] {};
    \node[left=.6cm, font=\boldmath] at (0, 3.5) {$C_{1}$};

    \node[fit=(v3) (v5), thick, dashed, draw=lightgray, minimum width=1.2cm] {};
    \node[left=.6cm, font=\boldmath, opacity=1] at (v4) {$C_{2}$};

    \node[fit=(u1) (u2), thick, dashed, draw=lightgray, opacity=0.8] {};
    \node[right=.25cm, font=\boldmath] at (2, 3.5) {$I_{2}$};
    
    \node[fit=(u3) (u5), thick, draw=gray] {};
    \node[right=.25cm, font=\boldmath] at (u4) {$I_{1}$};

    \draw[-] (v1) -- (v2);
    \draw[-] (v2) -- (v3);
    \draw[-] (v3) -- (v4);
    \draw[-] (v4) -- (v5);    
    \draw[-] (v1) to[bend right] (v3);
    \draw[-] (v1) to[bend right] (v4);
    \draw[-] (v1) to[bend right] (v5);
    \draw[-] (v2) to[bend right] (v4);
    \draw[-] (v2) to[bend right] (v5);
    \draw[-] (v3) to[bend right] (v5);

    \draw[-] (u1) -- (v1);
    \draw[-] (u2) -- (v1);
    \draw[-] (u2) -- (v2);
    \draw[-] (u2) -- (v3);
    \draw[-] (u3) -- (v4);
    \draw[-] (u3) -- (v5);
    \draw[-] (u4) -- (v5);
    \draw[-] (u5) -- (v3);
    \draw[-] (u5) -- (v5);  
\end{tikzpicture}